\documentclass[10pt]{article}
\usepackage{}
\usepackage{amssymb}
\usepackage{amsfonts}
\usepackage{bbm}
\usepackage{mathrsfs}
\usepackage{mathrsfs}
\usepackage{longtable}
\usepackage{supertabular}
\usepackage{amsfonts,mathrsfs,latexsym,amsmath,amssymb,amsthm}

\newtheorem{proposition}{Proposition}
\newtheorem{lemma}{Lemma}
\newtheorem{corollary}{Corollary}
\newtheorem{definition}{Definition}
\newtheorem{example}{Example}

\begin{document}
\title{\bf{On double cyclic codes over $\mathbb{Z}_4$}}
\author{{\bf  Jian Gao$^1$, Minjia Shi$^2$, Tingting Wu$^{1*}$, Fang-Wei Fu$^1$}\\
 {\footnotesize \emph{ 1. Chern Institute of Mathematics and LPMC, Nankai University}}\\
  {\footnotesize  \emph{Tianjin, 300071, P. R. China}}\\
   {\footnotesize \emph{ 2. School of Mathematical Sciences, Anhui University}}\\
  {\footnotesize  \emph{Hefei, 230601, P. R. China}}\\
  {\footnotesize  \emph{$^*$Corresponding author: wutingting@mail.nankai.edu.cn}}}
\date{}

\maketitle \noindent {\small {\bf Abstract} Let $R=\mathbb{Z}_4$ be the integer ring mod $4$. A double cyclic code of length $(r,s)$ over $R$ is a set that can be partitioned into two parts that any cyclic shift of the coordinates of both parts leaves invariant the code. These codes can be viewed as $R[x]$-submodules of $R[x]/(x^r-1)\times R[x]/(x^s-1)$. In this paper, we determine the generator polynomials of this family of codes as $R[x]$-submodules of $R[x]/(x^r-1)\times R[x]/(x^s-1)$. Further, we also give the minimal generating sets of this family of codes as $R$-submodules of $R[x]/(x^r-1)\times R[x]/(x^s-1)$. Some optimal or suboptimal nonlinear binary codes are obtained from this family of codes. Finally, we determine the relationship of generators between the double cyclic code and its dual. }
 \vskip 1mm



\vskip 3mm \baselineskip 0.2in

\section{Introduction}
Error-Correcting codes play important roles in applications ranging from data networking to satellite communication to compact disks. Classical coding theory concerns on linear codes since they have clear structure that makes them easy to discover, to understand and to encode and decode.
\par
Codes over finite rings have been studied since the early 1970s. There are a lot of works on codes over finite rings after the discovery that certain good nonlinear binary codes can be constructed from cyclic codes over $\mathbb{Z}_4$ via the Gray map \cite{Hammons}. Since then, many researchers have payed more and more attentions to study the codes over finite rings. In these studies, the group rings associated with codes are finite chain rings. Recently, the Spanish coding group Borges \emph{et. al}. introduced a class of new codes called $\mathbb{Z}_2\mathbb{Z}_4$-additive codes (see \cite{Borges1}). This family of codes are important from theory to application, and some generalizations are also studied deeply these years (see \cite{Abualrub2}, \cite{Aydogdu1}, \cite{Aydogdu2}). For $\mathbb{Z}_2\mathbb{Z}_4$-additive codes, the set of coordinates is partitioned into two parts, the first one of binary coordinates and the last one of quaternary coordinates. The generator matrices and duality of $\mathbb{Z}_2\mathbb{Z}_4$-additive codes were studied (see \cite{Borges1})£¬ and $\mathbb{Z}_2\mathbb{Z}_4$-additive cyclic codes were first studied by Abualrub \emph{et. al.} (see \cite{Abualrub1}). The methods given in \cite{Abualrub1} and \cite{Borges1} have been used efficiently to study some generalizations of $\mathbb{Z}_2\mathbb{Z}_4$-additive codes \cite{Abualrub2,Aydogdu1,Aydogdu2}.
\par
More recently, Borges \emph{et. al}. studied the algebraic structures of $\mathbb{Z}_2$-double cyclic codes (see \cite{Borges2}). In that literature, the authors determined the generator polynomials of this family of codes and their duals. In fact, the double cyclic codes were generalized quasi-cyclic (GQC) codes with index two introduced by Siap and Kulhan \cite{Siap} and studied deeply by many other researchers \cite{Cao1,Cao2,Esmaeili,Gao} actually. However, the points of view between \cite{Borges2} and \cite{Cao1,Cao2,Esmaeili,Gao,Siap} are different.
 \par
An important motivation to study linear codes, for example cyclic codes and their generalizations, over $\mathbb{Z}_4$ is that some good nonlinear binary codes can be obtained by these codes (see \cite{Aydin} \cite{Hammons}). Naturally, as a special class of linear codes, we ask that if some good nonlinear binary codes can be obtained from the double cyclic codes over $\mathbb{Z}_4$. In this paper, following the approaches given in \cite{Abualrub1} and \cite{Borges2}, we investigate some algebraic structures of double cyclic codes over $\mathbb{Z}_4$.  Some good nonlinear binary codes are obtained from this family of codes indeed.
\par
The paper is organized as follows. In Section 2, we introduce some definitions and give some structural properties of double cyclic codes over $\mathbb{Z}_4$. In Section 3, we determine the minimal generating sets of double cyclic codes over $\mathbb{Z}_4$. Moreover, some good nonlinear binary codes are obtained from this family of codes. In Section 4, we determine the relationship of generators between the double cyclic code and its dual.
\section{Double cyclic codes over $R$}
\label{sec:1}
Let $R=\mathbb{Z}_4$ be the integer ring mod $4$, then $R$ is a finite chain ring with the maximal ideal $(2)$ and the characteristic $4$. Let $R^n$ be the free $R$-module with rank $n$. A subset $\mathscr{C}$ of $R^n$ is called a linear code if and only if $\mathscr{C}$ is an $R$-submodule of $R^n$. Suppose that there are two odd positive integers $r, s$ such that $n=r+s$, similar to Definitions 2.1 and 2.2 given in \cite{Borges2}, we can also give the following definition of double cyclic codes over $R$.

\begin{definition}
Let $\mathscr{C}$ be a linear code of length $n$ over $R$. The code $\mathscr{C}$ is called a double cyclic code of length $(r, s)$ if for any
\begin{equation*}
c=(c_{1,0}, c_{1,1}, \ldots, c_{1, r-1}|c_{2,0}, c_{2,1},\ldots, c_{2,s-1})\in \mathscr{C}
\end{equation*}
implies the cyclic shift
\begin{equation*}
T(c)=(c_{1,r-1}, c_{1,0}, \ldots, c_{1, r-2}|c_{2,s-1}, c_{2,0},\ldots, c_{2,s-2})\in \mathscr{C}.
\end{equation*}
\end{definition}
\par
By the above definition, we can see that the double cyclic code $\mathscr{C}$ can be viewed as an $R$-submodule of $R^r\times R^s$.
\par
For any elements $$c=(c_{1,0}, c_{1,1}, \ldots, c_{1,r-1}|c_{2,0}, c_{2,1}, \ldots, c_{2,s-1})$$ and $$c'=(c'_{1,0}, c'_{1,1}, \ldots, c'_{1,r-1}|c'_{2,0}, c'_{2,1}, \ldots, c'_{2,s-1})$$
of $R^r\times R^s$, the inner product is defined as
 \begin{equation*}
 c\cdot c'=\sum_{i=0}^{r-1}c_{1,i}c'_{1,i}+\sum_{j=0}^{s-1}c_{2,j}c'_{2,j}.
 \end{equation*}
 Define the dual code of the double cyclic code $\mathscr{C}$ as
 \begin{equation*}
 \mathscr{C}^\perp=\{c'\in R^r\times R^s|c\cdot c'=0,~\forall c \in \mathscr{C}\}.
 \end{equation*}
 \begin{proposition}
 If $\mathscr{C}$ is a double cyclic code of length $(r,s)$ over $R$, then the dual code $\mathscr{C}^\perp$ is also a double cyclic code of length $(r, s)$ over $R$.
 \end{proposition}
 \begin{proof}
 Let $\mathscr{C}$ be a double cyclic code of length $(r, s)$ over $R$. Suppose that the element $c'=(c'_{1,0}, c'_{1,1}, \ldots, c'_{1,r-1}|c'_{2,0},c'_{2,1},\ldots,c'_{2,s-1})$ be a codeword of $\mathscr{C}^\bot$. It is sufficient to show that $T(c')\in \mathscr{C}^\perp$. Since $c'$ is a codeword of $\mathscr{C}^\perp$, it follows that for any codeword $c=(c_{1,0}, c_{1,1}, \ldots, c_{1,r-1}|c_{2,0}, c_{2,1}, \ldots, c_{2,s-1})\in \mathscr{C}$, we have $c'\cdot c=0$. We just need to show that $T(c')\cdot c=0$. Let $k={\rm lcm}(r,s)$, then $T^{k}(c)=c$ for any codeword $c$ of $\mathscr{C}$. Let $u=T^{k-1}(c)$. Since $\mathscr{C}$ is a double cyclic code, it follows that $T^{k-1}(c)=u \in \mathscr{C}$, then we have
 \begin{equation*}
 \begin{split}
 c'\cdot u&=c'_{1,0}c_{1,1}+\cdots+c'_{1,r-2}c_{1,r-1}+c'_{1,r-1}c_{1,0}+c'_{2,0}c_{2,1}+\cdots+c'_{2,s-2}c_{2,s-1}+c'_{2,s-1}c_{2,0}\\
          &=c_{1,0}c'_{1,r-1}+c_{1,1}c'_{1,0}+\cdots+c_{1,r-1}c'_{1,r-2}+c_{2,0}c'_{2,r-1}+c_{2,1}c'_{2,0}+\cdots+c_{2,r-1}c'_{2,r-2}\\
          &=c\cdot T(c')\\
          &=0,
  \end{split}
  \end{equation*}
which implies that $T(c')\in \mathscr{C}^\perp$. Thus, $\mathscr{C}^\perp$ is also a double cyclic code of length $(r, s)$ over $R$.
 \end{proof}
 \par
 Denote $R_{r, s}$ be the ring $R[x]/(x^r-1)\times R[x]/(x^s-1)$. Define a map $\tau$ from $R^r\times R^s$ to $R_{r, s}$ as $\tau((c_{1,0}, c_{1,1}, \ldots, c_{1, r-1}|c_{2,0}, c_{2,1},\ldots, c_{2,s-1}))=(c_{1,0}+c_{1,1}x+\cdots+c_{1,r-1}x^{r-1}| c_{2,0}+c_{2,1}x+\cdots+c_{2,s-1}x^{s-1})$, then $\tau$ is a bijective $R$-module isomorphism. Consider the following multiplication \begin{equation*}
x*(f(x)|g(x))=(xf(x)|xg(x)),
\end{equation*}
then the multiplication $*$ is well-defined, and $R_{r, s}$ is an $R[x]$-module with respect to this multiplication. Further, for any $c(x)=(c_{1,0}+c_{1,1}x+\cdots+c_{1,r-1}x^{r-1}| c_{2,0}+c_{2,1}x+\cdots+c_{2,s-1}x^{s-1})\in R_{r,s}$, $x*c(x)=(c_{1,r-1}+c_{1,0}x+\cdots+c_{1,r-2}x^{r-1}|c_{2,s-1}+c_{2,0}x+\cdots+c_{2,s-2}x^{s-1})$. Therefore,$x*c(x)$ is the image of the vector $(c_{1,r-1}, c_{1,0}, \ldots, c_{1, r-2}|c_{2,s-1}, c_{2,0},\ldots, c_{2,s-2})$, which implies that multiply the element $c(x)\in R_{r,s}$ by $x$ corresponding to a cyclic shift of the preimage of $c(x)$. Thus, the double cyclic code of length $(r, s)$ over $R$ can be viewed as an $R[x]$-submodule of $R_{r,s}$. In the following, we identity the double cyclic codes of length $(r, s)$ with the $R[x]$-submodules of $R_{r, s}$.
\begin{proposition}
Let $\mathscr{C}$ be a double cyclic code of length $(r, s)$ over $R$, then
\begin{equation*}
\mathscr{C}=((f_1(x)+2g_1(x)|0), (l(x)|f_2(x)+2g_2(x))),
\end{equation*}
where the polynomials $f_1(x), g_1(x), l(x), f_2(x), g_2(x)\in R[x]$ with $g_1(x)|f_1(x)|(x^r-1)$ and $g_2(x)|f_2(x)|(x^s-1)$.
\end{proposition}
\begin{proof}
Since the double cyclic code $\mathscr{C}$ and $R[x]/(x^s-1)$ are $R[x]$-submodules of $R_{r,s}$, we can define a map $\phi:~\mathscr{C}\rightarrow R[x]/(x^s-1)$ with $\phi((c_1(x)|c_2(x)))=c_2(x)$. Clearly, $\phi$ is an $R[x]$-module homomorphism and its image is an ideal of $R[x]/(x^s-1)$. Since $s$ is an odd positive integer, it follows that, by the theory of cyclic codes over $R$ (see \cite{Wan}), $\phi(\mathscr{C})=(f_2(x)+2g_2(x))$ with $f_2(x), g_2(x)\in R[x]$ and $g_2(x)|f_2(x)|(x^s-1)$. Note that
\begin{equation*}
Ker(\phi)=\{(c_1(x)|0)\in \mathscr{C}|~c_1(x)\in R[x]/(x^r-1)\}.
\end{equation*}
Define a set
\begin{equation*}
\mathcal {I}=\{c_1(x)\in R[x]/(x^r-1)|~(c_1(x)|0)\in Ker(\mathscr{C})\}.
\end{equation*}
Clearly, $\mathcal {I}$ is an ideal of $R[x]/(x^r-1)$. Therefore, there exist polynomials $f_1(x), g_1(x)\in R[x]/(x^r-1)$ with $g_1(x)|f_1(x)|(x^r-1)$ such that $\mathcal {I}=(f_1(x)+2g_1(x))$. Now, for any element $(c_1(x)|0)\in Ker(\phi)$, we have $c_1(x)\in \mathcal {I}$ and there exists some polynomial $m(x)\in R[x]$ such that $c_1(x)=m(x)(f_1(x)+2g_1(x))$. Thus
\begin{equation*}
(c_1(x)|0)=m(x)*(f_1(x)+2g_1(x)|0),
\end{equation*}
which implies that $Ker(\phi)$ is an $R[x]$-submodule of $\mathscr{C}$ generated by $(f_1(x)+2g_1(x)|0)$. Thus, by the first isomorphism theorem, we have
\begin{equation*}
\mathscr{C}/Ker(\phi)\cong (f_2(x)+2g_2(x)).
\end{equation*}
Let $(l(x)|f_2(x)+2g_2(x))\in \mathscr{C}$ with $\phi((l(x)|f_2(x)+2g_2(x)))=f_2(x)+2g_2(x)$, then any double cyclic code of length $(r, s)$ over $R$ can be generated as an $R[x]$-submodule of $R_{r, s}$ by two elements of the form $(f_1(x)+2g_1(x)|0)$ and $(l(x)|f_2(x)+2g_2(x))$ with $f_1(x), g_1(x), l(x), f_2(x), g_2(x)\in R[x]$ and $g_1(x)|f_1(x)|(x^r-1)$, $g_2(x)|f_2(x)|(x^s-1)$.
\end{proof}
\begin{lemma}
If $\mathscr{C}=((f_1(x)+2g_1(x)|0), (l(x)|f_2(x)+2g_2(x)))$ is a double cyclic code of length $(r, s)$ over $R$, then we may assume that ${\rm deg}(l(x))< {\rm deg}(f_1(x)+2g_1(x))$.
\end{lemma}
\begin{proof}
Suppose that ${\rm deg}(l(x))\geq {\rm deg}(f_1(x)+2g_1(x))$ with ${\rm deg}(l(x))- {\rm deg}(f_1(x)+2g_1(x))=i$. Consider another double cyclic code of length $(r, s)$ with generators of
\begin{equation*}
\mathcal {C}=((f_1(x)+2g_1(x)|0), (l(x)+x^i(f_1(x)+2g_1(x))|f_2(x)+2g_2(x))).
\end{equation*}
Clearly, $\mathcal {C}\subseteq \mathscr{C}$. However, we also have that $(l(x)|f_2(x)+2g_2(x))=(l(x)+x^i(f_1(x)+2g_1(x))|f_2(x)+2g_2(x))-x^i*(f_1(x)+2g_1(x)|0)$, which implies that $(l(x)|f_2(x)+2g_2(x))\in \mathscr{C}$. Therefore, $\mathscr{C}\subseteq \mathcal {C}$ implying $\mathscr{C}=\mathcal {C}$.
\end{proof}
\begin{lemma}
If $\mathscr{C}=((f_1(x)+2g_1(x)|0), (l(x)|f_2(x)+2g_2(x)))$ is a double cyclic code of length $(r, s)$ over $R$, then we may assume that $(f_1(x)+2g_1(x))|\frac{x^s-1}{g_2(x)}l(x)$.
\end{lemma}
\begin{proof}
Since $\frac{x^s-1}{g_2(x)}*(l(x)|f_2(x)+2g_2(x))=(\frac{x^s-1}{g_2(x)}l(x)|0)$, it follows that $\phi(\frac{x^s-1}{g_2(x)}*(l(x)|f_2(x)+2g_2(x)))=0$. Therefore, $(\frac{x^s-1}{g_2(x)}l(x)|0) \in Ker(\phi) \subseteq \mathscr{C}$ and $(f_1(x)+2g_1(x))|\frac{x^s-1}{g_2(x)}l(x)$.
\end{proof}
\par
Lemma 2 shows that if the double cyclic code $\mathscr{C}$ has only one generator of the form $(l(x)|f_2(x)+2g_2(x))$, then we have $(x^r-1)|\frac{x^s-1}{g_2(x)}l(x)$ and $g_2(x)|f_2(x)|(x^s-1)$. Thus, from this discussion and Lemmas 1 and 2, we have the following results directly.
\begin{proposition}
Let $\mathscr{C}$ be a double cyclic code of length $(r, s)$ over $R$, then we can classify $\mathscr{C}$ as follows:\\
{\rm (i)}~$\mathscr{C}=(f_1(x)+2g_1(x)|0)$ with $g_1(x)|f_1(x)|(x^r-1)$;\\
{\rm (ii)}~$\mathscr{C}=(l(x)|f_2(x)+2g_2(x))$ with $(x^r-1)|\frac{x^s-1}{g_2(x)}l(x)$ and $g_2(x)|f_2(x)|(x^s-1)$;\\
{\rm (iii)}~$\mathscr{C}=((f_1(x)+2g_1(x)|0), (l(x)|f_2(x)+2g_2(x)))$ with $g_1(x)|f_1(x)|(x^r-1)$, ${\rm deg}(l(x))< {\rm deg}(f_1(x)+2g_1(x))$, $(f_1(x)+2g_1(x))|\frac{x^s-1}{g_2(x)}l(x)$ and $g_2(x)|f_2(x)|(x^s-1)$.
\end{proposition}
\section{Minimal generating sets}
Let $\mathscr{C}$ be a double cyclic code of length $(r, s)$ over $R$. $\mathscr{C}$ is also an $R$-module. In this section, we will determine the generating sets of $\mathscr{C}$ in $R_{r, s}$ as an $R$-module. These sets will be used to determine the size and the generator matrix of $\mathscr{C}$.
\begin{proposition}
Let $\mathscr{C}=((f_1(x)+2g_1(x)|0), (l(x)|f_2(x)+2g_2(x)))$ be a double cyclic code of length $(r, s)$ over $R$ with ${\rm deg}(f_1(x))=t_1$, ${\rm deg}(g_1(x))=t_2$, $h_1(x)=\frac{x^r-1}{f_1(x)}$, ${\rm deg}(f_2(x))=r_1$, ${\rm deg}(g_2(x))=r_2$ and $h_2(x)=\frac{x^s-1}{f_2(x)}$. Let
\begin{equation*}
S_1=\bigcup_{i=0}^{r-t_1-1}\{ x^i*(f_1(x)+2g_1(x)|0)\},
\end{equation*}
\begin{equation*}
S_2=\bigcup_{i=0}^{t_1-t_2-1}\{ x^i*(2h_1(x)g_1(x)|0)\},
\end{equation*}
\begin{equation*}
S_3=\bigcup_{i=0}^{s-r_1-1}\{ x^i*(l(x)|f_2(x)+2g_2(x))\},
\end{equation*}
\begin{equation*}
S_4=\bigcup_{i=0}^{r_1-r_2-1}\{ x^i*(h_2(x)l(x)|2h_2(x)g_2(x))\},
\end{equation*}
then $S_1\cup S_2 \cup S_3 \cup S_4$ forms a minimal generating set for $\mathscr{C}$ as an $R$-submodule of $R_{r,s}$. Moreover, $\mathscr{C}$ has $4^{r+s-t_1-r_1}2^{t_1+r_1-t_2-r_2}$ codewords.
\end{proposition}
\begin{proof}
Let $c(x)$ be a codeword of $\mathscr{C}$, then there are polynomials $p(x)$ and $q(x)$ in $R[x]$ such that
\begin{equation*}
c(x)=p(x)*(f_1(x)+2g_1(x)|0)+q(x)*(l(x)|f_2(x)+2g_2(x)).
\end{equation*}
If ${\rm deg}(p(x))\leq r-t_1-1$, then $p(x)*(f_1(x)+2g_1(x)|0)\in {\rm Span}(S_1\cup S_2)$. Otherwise, by the division algorithm, there are two polynomials $p_1(x)$ and $t_1(x)$ in $R[x]$ such that
\begin{equation*}
p(x)=p_1(x)h_1(x)+t_1(x),
\end{equation*}
where $t_1(x)=0$ or ${\rm deg}(t_1(x))\leq r-t_1-1$. Therefore, we have that
\begin{equation*}
\begin{split}
p(x)*(f_1(x)+2g_1(x)|0)&= (h_1(x)p_1(x)+t_1(x))*(f_1(x)+2g_1(x)|0)\\
                       &= (2h_1(x)p_1(x)g_1(x)|0)+(t_1(x)(f_1(x)+2g_1(x))|0).
 \end{split}
 \end{equation*}
If ${\rm deg}(p_1(x))\leq t_1-t_2-1$, then we are done. Otherwise, there are polynomials $p_2(x)$ and $t_2(x)$ in $R[x]$ such that
\begin{equation*}
p_1(x)=\frac{x^r-1}{h_1(x)g_1(x)}p_2(x)+r_2(x),
\end{equation*}
where $r_2(x)=0$ or ${\rm deg}(r_2(x))\leq t_1-t_2-1$, then we have
\begin{equation*}
\begin{split}
p_1(x)*(2h_1(x)g_1(x)|0)&= (\frac{x^r-1}{h_1(x)g_1(x)}p_2(x)+r_2(x))*(2h_1(x)g_1(x)|0)\\
                        &= (2t_2(x)h_1(x)g_1(x)|0) \in {\rm Span}(S_2).
 \end{split}
 \end{equation*}
Thus, we have that $p(x)*(f_1(x)+2g_1(x)|0)\in {\rm Span}(S_1\cup S_2)$.
\par
If ${\rm deg}(q(x))\leq s-r_1-1$, then $q(x)*(l(x)|f_2(x)+2g_2(x)) \in {\rm Span}(S_1\cup S_2 \cup S_3 \cup S_4)$. Otherwise, there are two polynomials $q_1(x)$ and $r_1(x)$ in $R[x]$ such that
\begin{equation*}
q(x)=h_2(x)q_1(x)+r_1(x),
\end{equation*}
where $r_1(x)=0$ or ${\rm deg}(r_1(x))\leq s-r_1-1$, then we have
\begin{equation*}
\begin{split}
q(x)*(l(x)|f_2(x)+2g_2(x))&=(h_2(x)q_1(x)+r_1(x))*(l(x)|f_2(x)+2g_2(x))\\
                          &=q_1(x)*(h_2(x)l(x)|2h_2(x)g_2(x))+r_1(x)*(l(x)|f_2(x)+2g_1(x)).
 \end{split}
 \end{equation*}
Note that $r_1(x)*(l(x)|f_2(x)+2g_1(x))\in {\rm Span}(S_3)$. In the following, we will prove that $q_1(x)*(h_2(x)l(x)|2h_2(x)g_2(x))\in {\rm Span}(S_1\cup S_2 \cup S_4)$. From Lemma 2, we have that $(f_1(x)+2g_1(x))|\frac{x^s-1}{g_2(x)}l(x)$, then there exists a polynomial $k(x)$ in $R[x]$ such that $\frac{x^s-1}{g_2(x)}l(x)=(f_1(x)+2g_1(x))k(x)$. If ${\rm deg}(q_1(x))\leq r_1-r_2-1$, then we are done. Otherwise, there are polynomials $q_2(x)$ and $t_2(x)$ in $R[x]$ such that
\begin{equation*}
q_1(x)=\frac{x^s-1}{h_2(x)g_2(x)}q_2(x)+r_2(x),
\end{equation*}
where $r_2(x)=0$ or ${\rm deg}(r_2(x))\leq t_1-t_2-1$, then we have
\begin{equation*}
q_1(x)*(h_2(x)l(x)|2h_2(x)g_2(x))=(\frac{x^s-1}{g_2(x)}q_2(x)l(x)|0)+r_2(x)*(h_2(x)l(x)|2h_2(x)g_2(x)).
\end{equation*}
Clearly, $r_2(x)*(h_2(x)l(x)|2h_2(x)g_2(x)) \in {\rm Span}(S_4)$. Moreover, since $\frac{x^s-1}{g_2(x)}=(f_1(x)+2g_1(x))k(x)$, it follows that $(\frac{x^s-1}{g_2(x)}q_2(x)l(x)|0) \in {\rm Span}(S_1\cup S_2)$. Therefore, we have completed the proof that the set $S_1\cup S_2 \cup S_3 \cup S_4$ is a generating set of $\mathscr{C}$. It is also obvious that the set $S_1\cup S_2 \cup S_3 \cup S_4$ is minimal in the sense that no element in $S_1\cup S_2 \cup S_3 \cup S_4$ is a linear combination of the other elements. Note that the set $S_1$ and $S_2$ will contribute $4^{r-t_1}$ codewords and $2^{t_1-t_2}$ codewords, respectively, while the set $S_3$ and  $S_4$ will contribute $4^{s-r_1}$ codewords and $2^{r_1-r_2}$ codewords, respectively.
\end{proof}
\par
By Propositions 3 and 4, we have the following corollaries directly.
\begin{corollary}
Let $\mathscr{C}=((f_1(x)+2g_1(x)|0))$ be a double cyclic code of length $(r, s)$ over $R$ with $g_1(x)|f_1(x)|(x^r-1)$ and ${\rm deg}(f_1(x))=t_1$, ${\rm deg}(g_1(x))=t_2$, $h_1(x)=\frac{x^r-1}{f_1(x)}$. Let
\begin{equation*}
S_1=\bigcup_{i=0}^{r-t_1-1}\{ x^i*(f_1(x)+2g_1(x)|0)\}
\end{equation*}
and
\begin{equation*}
S_2=\bigcup_{i=0}^{t_1-t_2-1}\{ x^i*(2h_1(x)g_1(x)|0)\},
\end{equation*}
then $S_1\cup S_2$ forms a minimal generating set for $\mathscr{C}$ as an $R$-submodule of $R_{r,s}$. Moreover, $\mathscr{C}$ is equivalent to a cyclic code $(f_1(x)+2g_1(x))$ of length $r$ over $R$, and $\mathscr{C}$ has $4^{r-t_1}2^{t_1-t_2}$ codewords.
\end{corollary}
\begin{corollary}
Let $\mathscr{C}=(l(x)|f_2(x)+2g_2(x))$ be a double cyclic code of length $(r,s)$ over $R$ with $(x^r-1)|\frac{x^s-1}{g_2(x)}l(x)$ and $g_2(x)|f_2(x)|(x^s-1)$. Let ${\rm deg}(f_2(x))=r_1$, ${\rm deg}(g_2(x))=r_2$ and $h_2(x)=\frac{x^s-1}{f_2(x)}$. Let
\begin{equation*}
S_3=\bigcup_{i=0}^{s-r_1-1}\{ x^i*(l(x)|f_2(x)+2g_2(x))\}
\end{equation*}
and
\begin{equation*}
S_4=\bigcup_{i=0}^{r_1-r_2-1}\{ x^i*(h_2(x)l(x)|2h_2(x)g_2(x))\},
\end{equation*}
then $S_3\cup S_4$ forms a minimal generating set for $\mathscr{C}$ as an $R$-submodule of $R_{r,s}$. Moreover, $\mathscr{C}$ has $4^{s-r_1}2^{r_1-r_2}$ codewords.
\end{corollary}
\begin{corollary}
Let $\mathscr{C}$ be a double cyclic code of length $(r,s)$ over $R$. \\
{\rm (i)}~If $\mathscr{C}=((f_1(x)+2g_1(x)|0))$ with $g_1(x)|f_1(x)|(x^r-1)$ and ${\rm deg}(f_1(x))=t_1$, ${\rm deg}(g_1(x))=t_2$, $h_1(x)=\frac{x^r-1}{f_1(x)}$, then any codeword of $\mathscr{C}$ is of the form
\begin{equation*}
c(x)=a(x)*(f_1(x)+2g_1(x)|0)+b(x)*(2h_1(x)g_1(x)|0),
\end{equation*}
where $a(x)\in R[x]$, $b(x)\in \mathbb{F}_2[x]$ are polynomials with degrees $r-t_1-1$ and $t_1-t_2-1$ respectively.\\
{\rm (ii)}~If $\mathscr{C}=((l(x)|f_2(x)+2g_2(x)))$ with $(x^r-1)|\frac{x^s-1}{g_2(x)}l(x)$ and $g_2(x)|f_2(x)|(x^s-1)$, ${\rm deg}(f_2(x))=r_1$, ${\rm deg}(g_2(x))=r_2$, $h_2(x)=\frac{x^s-1}{f_2(x)}$, then any codeword of $\mathscr{C}$ is of the form
\begin{equation*}
c(x)=a(x)*(l(x)|f_2(x)+2g_2(x))+b(x)*(h_2(x)l(x)|2h_2(x)g_2(x)),
\end{equation*}
where $a(x)\in R[x]$, $b(x)\in \mathbb{F}_2[x]$ are polynomials with degrees $s-r_1-1$ and $r_1-r_2-1$ respectively.\\
{\rm (iii)}~If $\mathscr{C}=((f_1(x)+2g_1(x)|0), (l(x)|f_2(x)+2g_2(x)))$ with $g_1(x)|f_1(x)|(x^r-1)$, ${\rm deg}(l(x))< {\rm deg}(f_1(x)+2g_1(x))$, $(f_1(x)+2g_1(x))|\frac{x^s-1}{g_2(x)}l(x)$ and $g_2(x)|f_2(x)|(x^s-1)$. Let ${\rm deg}(f_1(x))=t_1$, ${\rm deg}(g_1(x))=t_2$, ${\rm deg}(f_2(x))=r_1$, ${\rm deg}(g_2(x))=r_2$, $h_1(x)=\frac{x^r-1}{f_1(x)}$ and $h_2(x)=\frac{x^s-1}{f_2(x)}$, then any codeword of $\mathscr{C}$ is of the form
\begin{equation*}
\begin{split}
  c(x)=&a(x)*(f_1(x)+2g_1(x)|0)+b(x)*(2h_1(x)g_1(x)|0)+u(x) \\
       &*(l(x)|f_2(x)+2g_2(x))+v(x)*(h_2(x)l(x)|2h_2(x)g_2(x)),
 \end{split}
 \end{equation*}
where $a(x), u(x) \in R[x]$ and $b(x), v(x) \in \mathbb{F}_2[x]$ are polynomials with degrees $r-t_1-1$, $s-r_1-1$, $t_1-t_2-1$ and $r_1-r_2-1$ respectively.
\end{corollary}

Define the Lee weights of the elements $0, 1, 2, 3$ of $R=\mathbb{Z}_4$ as $0, 1, 2, 1$, respectively. Moreover, the Lee weight of an $n$-tuple in $R^n$ is the sum of the Lee weights of its components. The Gray map $\Phi$ sends the elements $0,1,2,3$ of $R$ to $(0,0), (0,1), (1,1), (1,0)$ over $\mathbb{F}_2$, respectively. It is easy to verify that $\Phi$ is a distance-preserving map from ($R^n$, Lee distance) to ($\mathbb{F}_2^{2n}$, Hamming distance), but it is not an additive group homomorphism from $R^n$ to $\mathbb{F}_2^{2n}$ (see Chapter 3 in \cite{Wan}). Compared to the linear codes in the Code Tables \cite{Grassl} and nonlinear binary codes on the Table of Nonlinear Binary Codes \cite{Litsyn}, following examples computed by the computational algebra system Magma \cite{Bosma} show that some optimal or suboptimal nonlinear binary codes can be obtained from double cyclic codes over $R$.

\begin{example}
Let $\mathscr{C}=(l(x)|f_2(x)+2g_2(x))$ be a double cyclic code of length $(1, 7)$ over $R$ with $l(x)=1$ and $f_2(x)=g_2(x)=x^3+2x^2+x+3$ over $R$, then $S_3=\bigcup_{i=0}^3\{x^i*(1|3x^3+2x^2+3x+1)\}$ forms the generating set of $\mathscr{C}$, and $\mathscr{C}$ has $4^4$ codewords. Thus, as an $R$-submodule of $R_{1,7}$, $\mathscr{C}$ has the following matrix $G$ as its generator matrix
\begin{equation*} G=\left(
  \begin{array}{cccccccc}
    1 & 1 & 3& 2 & 3 & 0 & 0 & 0\\
    1 & 0 & 1& 3 & 2 & 3 & 0 & 0\\
    1 & 0 & 0& 1 & 3 & 2 & 3 & 0\\
    1 & 0 & 0& 0 & 1 & 3 & 2 & 3\\
  \end{array}
\right).
\end{equation*}
Further, the minimum Lee distance of $\mathscr{C}$ is $6$, which implies that $\Phi(\mathscr{C})$ is a $(16, 2^8, 6)$ nonlinear binary code, which is an optimal nonlinear binary code (see \cite{Litsyn}). The Lee weight enumerator is
\begin{equation*}
x^{16}+112x^2y^6+30x^8y^8+112x^2y^{10}+y^{16}.
\end{equation*}
In fact, the double cyclic code $\mathscr{C}$ given above is none other than the quaternary Kerdock code $\mathcal {K}(3)$ (see Chapter 8 in \cite{Wan}).
\end{example}

\begin{example}
Let $\mathscr{C}=(l(x)|f_2(x)+2g_2(x))$ be a double cyclic code of length $(1, 23)$ over $R$ with $l(x)=1$ and $f_2(x)=g_2(x)=x^{11}+3x^{10}+2x^7+x^6+x^5+x^4+x^2+2x+3$ over $R$, then $S_3=\bigcup_{i=0}^{11}\{x^i*(1|x^{11}+3x^{10}+2x^7+x^6+x^5+x^4+x^2+2x+3)\}$ forms the generating set of $\mathscr{C}$, and $\mathscr{C}$ has $4^{12}$ codewords. Further, the minimum Lee distance of $\mathscr{C}$ is $12$, which implies that $\Phi(\mathscr{C})$ is a nonlinear binary code $(48, 2^{24}, 12)$. This code gives the best parameters of nonlinear binary code (see \cite{Litsyn}). The Lee weight enumerator is
\begin{equation*}
\begin{split}
x^{48}&+12144x^{36}y^{12}+61824x^{34}y^{14}+195063x^{32}y^{16}+1133440x^{30}y^{18}+1445136x^{28}y^{20} \\
      &+4080384x^{26}y^{22}+2921232x^{24}y^{24}+4080384x^{22}y^{26}+1445136x^{20}y^{28}+1133440x^{18}y^{30}\\
      &+195063x^{16}y^{32}+61824x^{14}y^{34}+12144x^{12}y^{36}+y^{48}.
 \end{split}
 \end{equation*}

\end{example}

\begin{example}
Let $\mathscr{C}=(l(x)|f_2(x)+2g_2(x))$ be a double cyclic code of length $(3, 63)$ over $R$ with $l(x)=1+x+x^2$ and $f_2(x)=g_2(x)=x^{56} + 2x^{55} + 3x^{54} + 2x^{53} + 3x^{52} + 2x^{51} + 2x^{50} + 3x^{49} + x^{48} +x^{45} + 2x^{43} + x^{41} + 2x^{40} + 2x^{39} + x^{38} + x^{36} + 3x^{35} + 2x^{34} +3x^{33} + x^{32} + 2x^{31} + 3x^{28} + x^{27} + x^{26} + 2x^{25} + x^{24} + 2x^{22} +3x^{19} + 3x^{18} + x^{16} + x^{14} + x^{13} + 3x^{12} + 2x^{11} + 3x^9 + 3x^8 +3x^7 + 3x^6 + 3x^4 + 3x^3 + x^2 + x + 1$ over $R$, then $S_3=\bigcup_{i=0}^{6}\{x^i*(l(x)|3f_2(x))\}$ forms the generating set of $\mathscr{C}$, and $\mathscr{C}$ has $4^{7}$ codewords. Further, the minimum Lee distance of $\mathscr{C}$ is $56$, which implies that $\Phi(\mathscr{C})$ is a nonlinear binary code $(132, 2^{14}, 56)$. This code gives the parameters of the best known linear code with parameters $[132, 14, 56]$ (see \cite{Grassl}). The Lee weight enumerator is
\begin{equation*}
\begin{split}
x^{132}&+1260x^{76}y^{56}+2016x^{74}y^{58}+756x^{72}y^{60}+2079x^{58}y^{64}+4160x^{56}y^{66} \\
       &+2079x^{54}y^{68}+756x^{50}y^{72}+2016x^{48}y^{74}+1260x^{46}y^{76}+y^{132}.
 \end{split}
 \end{equation*}
\end{example}

\begin{example}
Let $\mathscr{C}=(l(x)|f_2(x)+2g_2(x))$ be a double cyclic code of length $(1, 15)$ over $R$ with $l(x)=1$ and $f_2(x)=g_2(x)=1+2x+x^2+2x^3+3x^5+3x^6+3x^8+x^9+x^{10}$ over $R$, then $S_3=\bigcup_{i=0}^{4}\{x^i*(1|1+2x+x^2+2x^3+3x^5+3x^6+3x^8+x^9+x^{10})\}$ forms the generating set of $\mathscr{C}$, and $\mathscr{C}$ has $4^{5}$ codewords. Further, the minimum Lee distance of $\mathscr{C}$ is $12$, which implies that $\Phi(\mathscr{C})$ is a nonlinear binary code $(32, 2^{10}, 12)$. This code has fewer codewords than the comparable best known nonlinear binary code (see \cite{Litsyn}). But this code has the parameters of best known binary linear code which is also optimal (see \cite{Grassl}). The Lee weight enumerator is
\begin{equation*}
x^{32}+240x^{20}y^{12}+542x^{16}y^{16}+240x^{12}y^{20}+y^{32}.
\end{equation*}
\end{example}

\section{Dual codes}
In this section, we determine the relationship of the generators between the double cyclic code and its dual. Let $\mathscr{C}=((f_1(x)+2g_1(x)|0), (l(x)|f_2(x)+2g_2(x)))$ be a double cyclic code of length $(r,s)$ with $g_1(x)|f_1(x)|(x^r-1)$ and $g_2(x)|f_2(x)|(x^s-1)$ over $R$. For simplicity, we denote the polynomials $f_1(x)+2g_1(x)$ and $f_2(x)+2g_2(x)$ by $F_1(x)$ and $F_2(x)$ respectively. In this section, we assume that the polynomials $F_1(x)$ and $F_2(x)$ are monic over $R$.
\par
From Proposition 1, we know that if $\mathscr{C}=((F_1(x)|0), (l(x)|F_2(x)))$ is a double cyclic code of length $(r, s)$ over $R$, then the dual code $\mathscr{C}^\perp$ is also a double cyclic code of length $(r,s)$ over $R$. We denote $\mathscr{C}^\perp=((\widehat{F}_1(x)|0), (\widehat{l}(x)|\widehat{F}_2(x)))$. Let $f(x)\in R[x]$ with degree $t$, then its reciprocal polynomial is denoted by $f^*(x)=x^tf(1/x)$. Further, we denote the polynomial $\sum_{i=0}^{m-1}x^i$ by $\theta_m(x)$. Let $k={\rm lcm}(r,s)$. By \cite[Definition 4.3]{Borges2}, we define the following map:

\begin{equation*}
\varphi:~R_{r,s}\times R_{r,s}\rightarrow R[x]/(x^k-1)
\end{equation*}
such that for any $c_1(x)=(c_{1,1}(x)|c_{1,2}(x))$ and $c_2(x)=(c_{2,1}(x)|c_{2,2}(x))$ of $R_{r,s}$, we have $\varphi((c_1(x), c_2(x)))=c_{1,1}(x)\theta_{\frac{k}{r}}(x^r)x^{k-1-{\rm deg}(c_{2,1}(x))}c^*_{2,1}(x)+c_{1,2}(x)\theta_{\frac{k}{s}}(x^s)x^{k-1-{\rm deg}(c_{2,2}(x))}c^*_{2,2}(x)$. The map $\varphi$ is a bilinear map between $R[x]$-modules.
\begin{lemma}
Let $c_1$ and $c_2$ be elements of $R^r\times R^s$ with associated polynomials $c_1(x)=(c_{1,1}(x)|c_{1,2}(x))$ and $c_2(x)=(c_{2,1}(x)|c_{2,2}(x))$ respectively, then $c_1$ is orthogonal to $c_2$ and all its cyclic shifts if and only if $\varphi((c_1(x), c_2(x)))=0$.
\end{lemma}
\begin{proof}
Let $$c_1=(c_{1,0},c_{1,1},\ldots,c_{1,r-1}|c'_{1,0},c'_{1,1},\ldots,c'_{1,s-1})$$ and $$c_2=(c_{2,0},c_{2,1},\ldots,c_{2,r-1}|c'_{2,0},c'_{2,1},\ldots,c'_{2,s-1}).$$
Let $$c^{(i)}_2=((c_{2,-i},c_{2,-i+1},\ldots, c_{2,r-1},c_{2,0}, \ldots, c_{2,-i-1})|(c'_{2,-i},c'_{2,-i+1},\ldots,c'_{2,s-1},c'_{2,0}, \ldots, c'_{2,-i-1}))$$ be the $i$-th cyclic shift of $c_2$, where $0\leq i \leq k-1$. Then $c_1 \cdot c^{(i)}_2=0$ if and only if $\sum_{j=0}^{r-1}c_{1,j}c_{2,j-i}+\sum_{n=0}^{s-1}c'_{1,n}c'_{2,n-i}=0$. Let $S_i=\sum_{j=0}^{r-1}c_{1,j}c_{2,j-i}+\sum_{n=0}^{s-1}c'_{1,n}c'_{2,n-i}$, then
\begin{equation*}
\begin{split}
\varphi (c_1(x), c_2(x))&=\sum_{u=0}^{r-1}\left(\theta_{\frac{k}{r}}(x^r)\sum_{j=0}^{r-1}c_{1,j}c_{2,j-u}x^{k-1-u}\right)+ \sum_{t=0}^{s-1}\left(\theta_{\frac{k}{s}}(x^s)\sum_{n=0}^{s-1}c'_{1,n}c'_{2,n-t}x^{k-1-t}\right) \\
                        &= \sum_{i=0}^{k-1}S_ix^{k-1-i}
 \end{split}
 \end{equation*}
in $R[x]/(x^k-1)$. Thus, $\varphi(c_1(x), c_2(x))=0$ if and only if $S_i=0$ for all $0\leq i \leq k-1$.
\end{proof}
\begin{lemma}
Let $(F_1(x)|0)$ and $(l(x)|F_2(x))$ belong to $R_{r,s}$, where $\varphi((F_1(x)|0), (l(x)|F_2(x)))=0$, then $F_1(x)l^*(x)=0$ in $R[x]/(x^r-1)$. Respectively, if $(0|D_1(x))$ and $(m(x)|D_2(x))$ are elements of $R_{r,s}$ with the condition $\varphi((0|D_1(x)),(m(x)|D_2(x)))=0$, then $D_1(x)D^*_2(x)=0$ in $R[x]/(x^s-1)$.
\end{lemma}
\begin{proof}
By the definition of the map $\varphi$, we have that
\begin{equation*}
\varphi((F_1(x)|0), (l(x)|F_2(x)))=F_1(x)\theta_{\frac{k}{r}}(x^r)x^{k-1-{\rm deg}(l(x))}l^*(x)=0
\end{equation*}
in $R[x]/(x^k-1)$, which implies that there exists a polynomial $f(x)\in R[x]$ such that $F_1(x)\theta_{\frac{k}{r}}(x^r)x^{k-1-{\rm deg}(l(x))}l^*(x)=f(x)(x^k-1)$. Suppose that $g(x)=f(x)x^{{\rm deg}(l(x))+1}$, then we have
\begin{equation*}
F_1(x)x^kl^*(x)=f(x)x^{{\rm deg}(l(x))+1}(x^r-1).
\end{equation*}
Since $x$ and $x^r-1$ are coprime to each other, it follows that $F_1(x)l^*(x)=0$ in $R[x]/(x^r-1)$.
The same argument can be used to prove the other case.
\end{proof}
\begin{proposition}
Let $\mathscr{C}=((F_1(x)|0), (l(x)|F_2(x)))$ be a double cyclic code of length $(r,s)$ over $R$. Let $\mathscr{C}^\perp=((\widehat{F}_1(x)|0), (\widehat{l}(x)|\widehat{F}_2(x)))$ be its dual code, then
\begin{equation*}
\widehat{F}^*_1(x){\rm gcd}(F_1(x), l(x))=\lambda(x)(x^r-1)
\end{equation*}
for some $\lambda(x)\in R[x]$.
\end{proposition}
\begin{proof}
Since the elements $(\widehat{F}_1(x)|0)$ and $(\widehat{l}(x)|\widehat{F}_2(x))$ belong to $\mathscr{C}^\perp$, it follows that, by Lemma 3, $\varphi((\widehat{F}_1(x)|0), (F_1(x)|0))=0$ and $\varphi((\widehat{l}(x)|\widehat{F}_2(x)), (F_1(x)|0))=0$ in $R[x]/(x^k-1)$. Therefore, by Lemma 4, we have that $F_1(x)\widehat{F}^*_1(x)=0$ and $l(x)\widehat{F}^*_1(x)=0$ in $R[x]/(x^r-1)$. It means that ${\rm gcd}(F_1(x), l(x))\widehat{F}^*_1(x)=0$ in $R[x]/(x^r-1)$. Thus, there exists a polynomial $\lambda(x)\in R[x]$ such that $\widehat{F}^*_1(x){\rm gcd}(F_1(x), l(x))=\lambda(x)(x^r-1)$.
\end{proof}
\begin{proposition}
Let $\mathscr{C}=((F_1(x)|0), (l(x)|F_2(x)))$ be a double cyclic code of length $(r,s)$ over $R$. Let $\mathscr{C}^\perp=((\widehat{F}_1(x)|0), (\widehat{l}(x)|\widehat{F}_2(x)))$ be its dual code, then
\begin{equation*}
\widehat{F}^*_2(x)F_1(x)F_2(x)=\mu(x)(x^s-1){\rm gcd}(F_2(x), l(x))
\end{equation*}
for some $\mu(x) \in R[x]$.
\end{proposition}
\begin{proof}
Let $$c(x)=\frac{l(x)}{{\rm gcd}(F_1(x), l(x))}*(F_1(x)|0)-\frac{F_1(x)}{{\rm gcd}(F_1(x), l(x))}*(l(x)|F_2(x))=\left(0|\frac{F_2(x)F_1(x)}{{\rm gcd}(F_1(x), l(x))}\right),$$ then $c(x)\in \mathscr{C}$, which implies that
\begin{equation*}
\varphi((0|F_2(x)\frac{F_1(x)}{{\rm gcd}(F_1(x), l(x))}), (\widehat{l}(x)| \widehat{F}_2(x)))=0
\end{equation*}
in $R[x]/(x^k-1)$. According to Lemma 4, there exists a polynomial $\mu(x)\in R[x]$ such that
$\widehat{F}^*_2(x)F_1(x)F_2(x)=\mu(x)(x^s-1){\rm gcd}(F_2(x), l(x))$.\end{proof}

\begin{proposition}
Let $\mathscr{C}=((F_1(x)|0), (l(x)|F_2(x)))$ be a double cyclic code of length $(r,s)$ over $R$. Let $\mathscr{C}^\perp=((\widehat{F}_1(x)|0), (\widehat{l}(x)|\widehat{F}_2(x)))$ be its dual code, then we have
\begin{equation*}
\widehat{l}^*(x)F_1(x)=\nu(x)(x^r-1)
\end{equation*}
for some $\nu(x) \in R[x]$.
\end{proposition}
\begin{proof}
Let $c(x)=(\widehat{F}_1(x)|0)+(\widehat{l}(x)|\widehat{F}_2(x))$, then $c(x)\in \mathscr{C}^\perp$. Since the map $\varphi$ is bilinear between $R[x]$-module, we have that
\begin{equation*}
\begin{split}
\varphi(c(x), (F_1(x)|0))&= \varphi((\widehat{F}_1(x)|0), (F_1(x)|0))+\varphi(\widehat{l}(x)|\widehat{F}_2(x), (F_1(x)|0))\\
                         &= \varphi(\widehat{l}(x)|\widehat{F}_2(x), (F_1(x)|0))\\
                         &=0
 \end{split}
 \end{equation*}
 in $R[x]/(x^k-1)$. In light of Lemma 4, there exists a polynomial $\nu(x)\in R[x]$ such that $F_1(x)\widehat{l}^*(x)=\nu(x)(x^r-1)$.
\end{proof}
\par
Define $\overline{\mathscr{C}}$ be the residue code of $\mathscr{C}$, i.e. $\overline{\mathscr{C}}=\{\overline{c}=c~{\rm mod}~2|~c\in \mathscr{C}\}$. Clearly, $\overline{\mathscr{C}}$ is a $\mathbb{Z}_2$-double cyclic code. If $\mathscr{C}=((F_1(x)|0), (l(x)|F_2(x)))$, then $\overline{\mathscr{C}}=((\overline{F}_1(x)|0), (\overline{l}(x)|\overline{F}_2(x)))$, where $\overline{f}(x)$ denotes the polynomial over $\mathbb{Z}_2$ with the coefficients of the polynomial $f(x)$ mod $2$. Therefore, by Corollaries 4.7, 4.8 and Proposition 4.18 in \cite{Borges2}, we have the following lemma directly.
\begin{lemma}
Let $\mathscr{C}=((F_1(x)|0), (l(x)|F_2(x)))$ be a double cyclic code of length $(r,s)$ over $\mathbb{Z}_4$ with $\mathscr{C}^\perp=((\widehat{F}_1(x)|0), (\widehat{l}(x)|\widehat{F}_2(x)))$, then we have
\begin{equation*}
{\rm deg}(\overline{\widehat{F}}_1(x))=r-{\rm deg}({\gcd}(\overline{F}_1(x),\overline{l}(x))),
\end{equation*}
\begin{equation*}
{\rm deg}(\overline{\widehat{F}}_2(x))=s-{\rm deg}(\overline{F}_2(x))-{\rm deg}(\overline{F}_1(x))+{\rm deg}({\gcd}(\overline{F}_1(x),\overline{l}(x))).
\end{equation*}
Further, let $A(x)=\frac{\overline{l}(x)}{{\rm gcd}(\overline{F}_1(x), \overline{l}(x))}$, then we have
\begin{equation*}
\left(\overline{\nu}(x)x^{k-{\rm deg}(\overline{l}(x))-1}A^*(x)+x^{k-{\rm deg}(\overline{F}_2(x))-1}\right)=0~{\rm mod}~ \left(\frac{\overline{F}^*_1(x)}{{\rm gcd}^*(\overline{F}_1(x),\overline{l}(x))}\right).
\end{equation*}
\end{lemma}
From Propositions 5, 6, 7 and Lemma 5, we have the following result.
\begin{corollary}
Let $\mathscr{C}=((F_1(x)|0), (l(x)|F_2(x)))$ be a double cyclic code of length $(r,s)$ over $R$. Let $\mathscr{C}^\perp=((\widehat{F}_1(x)|0), (\widehat{l}(x)|\widehat{F}_2(x)))$ be its dual code, then $\overline{\widehat{F_1}^*}(x)=\frac{x^r-1}{{\rm gcd}(\overline{F}_1(x), \overline{l}(x))}$ and $\overline{\widehat{F_2}^*}(x)=\frac{(x^s-1){\rm gcd}
(\overline{F}_2(x), \overline{l}(x))}{\overline{F}_1(x)\overline{F}_2(x)}$. Further,
\begin{equation*}
\overline{\nu}(x)=x^{k-{deg}(\overline{F}_2(x))+{\rm deg}(\overline{l}(x))}(A^*(x))^{-1}~{\rm mod}~ \left(\frac{\overline{F}^*_1(x)}{{\rm gcd}^*(\overline{F}_1(x),\overline{l}(x))}\right).
\end{equation*}
\end{corollary}
\begin{proof}
By Proposition 5, we have that
\begin{equation*}
\overline{\widehat{F_1}^*}(x)=\overline{\lambda}(x)\frac{x^r-1}{{\rm gcd}(\overline{F}_1(x), \overline{l}(x))}.
\end{equation*}
Since ${\rm deg}(\overline{\widehat{F_1}^*}(x))={\rm deg}(\overline{\widehat{F}}_1(x))$, by Lemma 5, it follows that $\overline{\lambda}(x)=1$. Therefore, we have that
\begin{equation*}
\overline{\widehat{F_1}^*}(x)=\frac{x^r-1}{{\rm gcd}(\overline{F}_1(x), \overline{l}(x))}.
\end{equation*}
Similarly, one can also prove $\overline{\mu}(x)=1$, i.e. $\overline{\widehat{F_2}^*}(x)=\frac{(x^s-1){\rm gcd}(\overline{F}_2(x), \overline{l}(x))}{\overline{F}_1(x)\overline{F}_2(x)}$. The proof of the last part is similar to that of Corollary 4.19 in \cite{Borges2}, and is omitted here.
\end{proof}

Let $\mathscr{C}=((F_1(x)|0), (l(x)|F_2(x)))$ be a double cyclic code of length $(r,s)$ over $R$, then, according to Proposition 4, $\mathscr{C}$ is a free $R$-submodule of $R_{r,s}$ if and only if $F_1(x)|(x^r-1)$ and $F_2(x)|(x^s-1)$. In the rest of this paper, following the approach given in \cite{Borges2}, we will determine the explicit relationship of the free double cyclic code $\mathscr{C}$ and its dual $\mathscr{C}^\perp$.
\par
Let $\mathscr{C}_r$ be the canonical projection of $\mathscr{C}$ on the first $r$ coordinates and $\mathscr{C}_s$ on the last $s$ coordinates. The canonical projection is a linear map. Therefore, $\mathscr{C}_r$ and $\mathscr{C}_s$ are cyclic codes of length $r$ and $s$ over $R$, respectively.
\begin{lemma}
Let $\mathscr{C}=((F_1(x)|0), (l(x)|F_2(x)))$ be a free double cyclic code of length $(r,s)$ over $R$, then
\begin{equation*}
|\mathscr{C}_r|=4^{r-{\rm deg}(F_1(x))+\varepsilon},~|\mathscr{C}_s|=4^{s-{\rm deg}(F_2(x))};
\end{equation*}
\begin{equation*}
|(\mathscr{C}_r)^\perp|=4^{{\rm deg}(F_1(x))-\varepsilon},~|(\mathscr{C}_s)^\perp|=4^{{\rm deg}(F_2(x))};
\end{equation*}
\begin{equation*}
|(\mathscr{C}^\perp)_r|=4^{{\rm deg}(F_1(x))},~|(\mathscr{C}^\perp)_s|=4^{{\rm deg}(F_2(x))+\varepsilon};
\end{equation*}
where $\varepsilon={\rm deg}(F_1(x))-{\rm deg}({\rm gcd}(F_1(x),l(x)))$.
\end{lemma}
\begin{proof}
Clearly, $\mathscr{C}_r=(F_1(x),l(x))=({\rm gcd}(F_1(x),l(x)))$. Since $F_1(x)|(x^r-1)$, it follows that ${\rm gcd}(F_1(x),l(x))|(x^r-1)$ and $|\mathscr{C}_r|=4^{r-{\rm gcd}(F_1(x),l(x))}$ and $|(\mathscr{C}_r)^\perp|=4^{{\rm deg}({\rm gcd}(F_1(x),l(x)))}$. The proof of $|\mathscr{C}_s|$ and $|(\mathscr{C}_s)^\perp|$ is similar to this. Since $\mathscr{C}$ is a free $R$-submodule of $R_{r,s}$, then $\mathscr{C}^\perp$ is also a free $R$-submodule of $R_{r,s}$. Therefore, $(\mathscr{C}^\perp)_r$ and $(\mathscr{C}^\perp)_s$ are the Hensel Lift of $\overline{(\mathscr{C}^\perp)_r}$ and $\overline{(\mathscr{C}^\perp)_s}$, respectively. From Corollary 4, we have that $\overline{(\mathscr{C}^\perp)_r}=({\rm gcd}(\overline{\widehat{F}}_1(x), \overline{\widehat{l}}(x)))=\left( \frac{x^r-1}{\overline{F}_1(x)}\right)$. Since $r$ is odd and $F_1(x)|(x^r-1)$, by the uniqueness of the Hensel Lift (see Chapter 5 in \cite{Wan}), we have that $\frac{x^r-1}{F_1(x)}$ is the Hensel Lift of $\frac{x^r-1}{\overline{F}_1(x)}$ over $R$. It means that $(\mathscr{C}^\perp)_r=\left( \frac{x^r-1}{F_1(x)}\right)$ and $|(\mathscr{C}^\perp)_r|=4^{{\rm deg}(F_1(x))}$. Similarly, we have $|(\mathscr{C}^\perp)_s|=4^{{\rm deg}(F_2(x))+\varepsilon}$.
\end{proof}
\begin{lemma}
Let $\mathscr{C}=((F_1(x)|0), (l(x)|F_2(x)))$ be a free double cyclic code of length $(r,s)$ over $R$. Let $\mathscr{C}^\perp=((\widehat{F}_1(x)|0), (\widehat{l}(x),\widehat{F}_2(x)))$, then
\begin{equation*}
{\rm deg}(\widehat{F}_1(x))=r-{\rm deg}({\rm gcd}(F_1(x), l(x)))
\end{equation*}
and
\begin{equation*}
{\rm deg}(\widehat{F}_2(x))=s-{\rm deg}(F_2(x))-{\rm deg}(F_1(x))+{\rm deg}({\rm gcd}(F_1(x), l(x))).
\end{equation*}
\end{lemma}
\begin{proof}
Clearly, $(\mathscr{C}_r)^\perp$ is a cyclic code generated by $\frac{x^r-1}{{\rm gcd}(F_1(x), l(x))}$, which is the Hensel Lift of $\widehat{\overline{F}}_1(x)=\frac{x^r-1}{{\rm gcd}(\overline{F}_1(x), \overline{l}(x))}$. Thus ${\rm deg}(\widehat{F}_1(x))={\rm deg}(\widehat{\overline{F}}_1(x))$. Since ${\rm gcd}(F_1(x), l(x))|(x^r-1)$, it follows that ${\rm gcd}(F_1(x), l(x))$ is the Hensel Lift of ${\rm gcd}(\overline{F}_1(x), \overline{l}(x))$, i.e. ${\rm deg}({\rm gcd}(F_1(x), l(x)))={\rm deg}({\rm gcd}(\overline{F}_1(x), \overline{l}(x)))$. Since $\mathscr{C}^\perp$ is a free $R$-submodule of $R_{r,s}$, it follows that $(\mathscr{C}^\perp)_s$ is a free cyclic code generated by $\widehat{F}_2(x)$, i.e. $|(\mathscr{C}^\perp)_s|=4^{s-{\rm deg}(\widehat{F}_2(x))}$. Moreover, by Lemma 6, $|(\mathscr{C}^\perp)_s|=4^{{\rm deg}(F_2(x))+\varepsilon}$. Therefore, we have that ${\rm deg}(\widehat{F}_2(x))=s-{\rm deg}(F_2(x))-{\rm deg}(F_1(x))+{\rm deg}({\rm gcd}(F_1(x), l(x)))$.\end{proof}

From Propositions 5, 6, 7 and Lemma 7, we have the following result.
\begin{corollary}
Let $\mathscr{C}=((F_1(x)|0), (l(x)|F_2(x)))$ be a free double cyclic code of length $(r,s)$ over $R$. Let $\mathscr{C}^\perp=((\widehat{F}_1(x)|0), (\widehat{l}(x)|\widehat{F}_2(x)))$ be its dual code, then $\widehat{F}^*_1(x)=\frac{x^r-1}{{\rm gcd}(F_1(x), l(x))}$ and $\widehat{F}^*_2(x)=\frac{(x^s-1){\rm gcd}(F_2(x), l(x))}{F_1(x)F_2(x)}$. Further,  let $A(x)=\frac{l(x)}{{\rm gcd}(F_1(x), l(x))}$, then
\begin{equation*}
\nu(x)=x^{k-{deg}(F_2(x))+{\rm deg}(l(x))}(A^*(x))^{-1}~{\rm mod}~ \left(\frac{F^*_1(x)}{{\rm gcd}^*(F_1(x),l(x))}\right).
\end{equation*}
\end{corollary}
\begin{proof}
The proof process is similar to that of Corollary 4.
\end{proof}
Finally, we summarize the results on dual codes in the following result.
\begin{proposition}
Let $\mathscr{C}=((F_1(x)|0), (l(x)|F_2(x)))$ be a double cyclic code of length $(r,s)$ over $R$, where $F_1(x)$ and $F_2(x)$ be monic polynomials over $R$. Let $\mathscr{C}^\perp=((\widehat{F}_1(x)|0), (\widehat{l}(x)|\widehat{F}_2(x)))$ be its dual code, then we have\\
{\rm (i)}~Let $A(x)=\frac{\overline{l}(x)}{{\rm gcd}(\overline{F}_1(x), \overline{l}(x))}$, then $\overline{\widehat{F_1}^*}(x)=\frac{x^r-1}{{\rm gcd}(\overline{F}_1(x), \overline{l}(x))}$ and $\overline{\widehat{F_2}^*}(x)=\frac{(x^s-1){\rm gcd}(\overline{F}_2(x), \overline{l}(x))}{\overline{F}_1(x)\overline{F}_2(x)}$. Further, $\widehat{l}^*(x)F_1(x)=\nu(x)(x^r-1)$ with
\begin{equation*}
\overline{\nu}(x)=x^{k-{deg}(\overline{F}_2(x))+{\rm deg}(\overline{l}(x))}(A^*(x))^{-1}~{\rm mod}~ \left(\frac{\overline{F_1}^*(x)}{{\rm gcd}^*(\overline{F}_1(x),\overline{l}(x))}\right).
\end{equation*}
{\rm (ii)}~Let $A(x)=\frac{l(x)}{{\rm gcd}(F_1(x), l(x))}$. If $\mathscr{C}$ is a free $R$-submodule of $R_{r,s}$, then $\widehat{F}^*_1(x)=\frac{x^r-1}{{\rm gcd}(F_1(x), l(x))}$ and $\widehat{F}^*_2(x)=\frac{(x^s-1){\rm gcd}(F_2(x), l(x))}{F_1(x)F_2(x)}$. Further, $\widehat{l}^*(x)F_1(x)=\nu(x)(x^r-1)$ with
\begin{equation*}
\nu(x)=x^{k-{deg}(F_2(x))+{\rm deg}(l(x))}(A^*(x))^{-1}~{\rm mod}~ \left(\frac{F^*_1(x)}{{\rm gcd}^*(F_1(x),l(x))}\right).
\end{equation*}
\end{proposition}
\begin{example}
Let $\mathscr{C}=((x^2+x+1|0), (x+1|x^6+x^3+1))$ be a free double cyclic code of length $(3,9)$ over $R$, then we have $F_1(x)=x^2+x+1$, $l(x)=x+1$ and $F_2(x)=x^6+x^3+1$. From Proposition $4$, we have that $\mathscr{C}$ has the following sets as its generating sets
\begin{equation*}
S_1=\{(x^2+x+1|0)\},~S_3=\bigcup_{i=0}^2\{x^i*(x+1|x^6+x^3+1)\}.
\end{equation*}
Moreover, $|\mathscr{C}|=4^4$. From Proposition 8, we have $\widehat{F}^*_1(x)=\frac{x^3-1}{{\rm gcd}(x^2+x+1,x+1)}=0$ in $R[x]/(x^3-1)$, and $\widehat{F}^*_2(x)=\frac{x^9-1}{(x^6+x^3+1)(x^2+x+1)}=x+3$ in $R[x]/(x^9-1)$. Since $A(x)=x+1$, it follows that $A^*(x)=x+1$ and $(A^*(x))^{-1}=3x~{\rm mod}~x^2+x+1$. Therefore, $\nu(x)=3x^{9-6+1}x=3x^5=x+1~{\rm mod}~(x^2+x+1)$, which implies that $\widehat{l}^*(x)=\frac{(x^3-1)(x+1)}{x^2+x+1}=x^2-1$. Thus, $\widehat{l}(x)=3x^2+1$. It means that $\mathscr{C}^\perp=(3x^2+1|3x+1)=(x^2-1|x-1)$. The generating set of $\mathscr{C}^\perp$ is $S_3=\bigcup_{i=0}^7\{ x^i*(x^2-1|x-1)\}$ and $|\mathscr{C}^\perp|=4^8$.
\end{example}
\section{Conclusion}
This paper is devoted to the study of double cyclic codes over $\mathbb{Z}_4$. We first determine the generator polynomials of this family of codes, and give their minimal generating sets.  Further, we also discuss the relationship of generators between the double cyclic code and its dual. Examples are given to show that some optimal or suboptimal nonlinear binary codes can be obtained from this family of codes. We believe that some more optimal or new nonlinear binary codes can be obtained from double cyclic codes over $\mathbb{Z}_4$, and it will be an interesting and challenging work in future.

\end{document}